\documentclass[journal,10pt,twocolumn]{IEEEtran}
\usepackage[table]{xcolor}
\usepackage{comment}
\usepackage{xparse}
\usepackage{mwe}
\usepackage{bbm}
\usepackage{multicol}
\usepackage{color, colortbl}
\usepackage{amsmath,amssymb,amsthm}

\usepackage{cite}
\usepackage{pdflscape}
\usepackage[english]{babel}
\usepackage{multicol}
\usepackage{multirow}
\newtheorem{theorem}{Theorem}

\usepackage{amsmath}
\usepackage{amsmath}
\usepackage{amsfonts}
\usepackage{amssymb}
\usepackage{float}
\usepackage{flushend}
\usepackage{stfloats}
\usepackage{subfigure}
\usepackage{float}
\usepackage{amsmath}
\usepackage{mathrsfs}
\usepackage{color}
\usepackage{float}
\usepackage{array}
\usepackage{balance}
\usepackage{graphicx}
\usepackage{epsfig}
\usepackage{epstopdf}
\usepackage{algorithm}
\usepackage{textcomp}
\usepackage[noend]{algorithmic}
\usepackage{url}
\usepackage{bm}
\usepackage{footnote}
\usepackage{comment}
\usepackage[numbers, square, comma, sort&compress]{natbib}
\usepackage[numbers]{natbib}

\usepackage[singlespacing]{setspace}
\usepackage{etoolbox}
\usepackage{ccaption}
\usepackage[font=footnotesize,labelfont=rm,figurename=Fig.]{caption}
\IEEEoverridecommandlockouts
\theoremstyle{definition}
\newtheorem{defn}{Definition}[]

\newcommand{\E}{\mathbb{E}}
\newcommand{\R}{\mathbb{R}}

\newcommand{\Expect}{{\rm I\kern-.9em E}}

\newcounter{inlineequation}
\DeclareMathAlphabet{\mathpzc}{OT1}{pzc}{m}{it}
\definecolor{Gray}{gray}{0.9}
\definecolor{LightCyan}{rgb}{0.88,1,1}

\NewDocumentCommand{\ceil}{s O{} m}{%
  \IfBooleanTF{#1} 
    {\left\lceil#3\right\rceil} 
    {#2\lceil#3#2\rceil} 
}

\NewDocumentCommand{\floor}{s O{} m}{%
  \IfBooleanTF{#1} 
    {\left\lfloor#3\right\rfloor} 
    {#2\lfloor#3#2\rfloor} 
}

\begin{document}
\title{Intracell Interference Characterization and \\Cluster Inference for D2D Communication}

\author{\IEEEauthorblockN{Hafiz Attaul Mustafa\IEEEauthorblockA{$^1$}, Muhammad Zeeshan Shakir\IEEEauthorblockA{$^2$}, Ali Riza Ekti\IEEEauthorblockA{$^3}$, \\Muhammad Ali Imran\IEEEauthorblockA{$^1$}, and Rahim Tafazolli\IEEEauthorblockA{$^1$}}\\
\vspace{0.1in}
\fontsize{11}{11}\selectfont
\IEEEauthorblockA{$^1$Institute for Communication Systems, University of Surrey, Guildford, UK.}\\
\vspace{0.05in}
\fontsize{11}{11}\selectfont\ttfamily\upshape
\{h.mustafa, m.imran, r.tafazolli\}@surrey.ac.uk \\
\vspace{0.1in}
\fontsize{11}{11}\selectfont\rmfamily\
\IEEEauthorblockA{$^2$Dept. of Systems and Computer Engineering, Carleton University, Ottawa, Canada.}\\
\vspace{0.05in}
\fontsize{11}{11}\selectfont\ttfamily\upshape
muhammad.shakir@sce.carleton.ca \\
\vspace{0.1in}
\fontsize{11}{11}\selectfont\rmfamily\
\IEEEauthorblockA{$^3$Dept. of Electrical and Computer Engineering, Gannon University, Erie, PA, U.S.A.\\
\vspace{0.05in}
\fontsize{11}{11}\selectfont\ttfamily\upshape
alirizaekti@gmail.com}\\
}
\maketitle
\begin{abstract}
The homogeneous poisson point process (PPP) is widely used to model temporal, spatial or both topologies of base stations (BSs) and mobile terminals (MTs). However, negative spatial correlation in BSs, due to strategical deployments, and positive spatial correlations in MTs, due to homophilic relations, cannot be captured by homogeneous spatial PPP (SPPP). In this paper, we assume doubly stochastic poisson process, a generalization of homogeneous PPP, with intensity measure as another stochastic process. To this end, we assume Permanental Cox Process (PCP) to capture positive spatial correlation in MTs. We consider product density to derive closed-form approximation (CFA) of spatial summary statistics. We propose Euler Characteristic (EC) based novel approach to approximate intractable random intensity measure and subsequently derive nearest neighbor distribution function. We further propose the threshold and spatial extent of excursion set of chi-square random field as interference control parameters to select different cluster sizes for device-to-device (D2D) communication. The spatial extent of clusters is controlled by nearest neighbor distribution function which is incorporated into Laplace functional of SPPP to analyze the effect of D2D interfering clusters on average coverage probability of cellular user. The CFA and empirical results are in good agreement and its comparison with SPPP clearly shows spatial correlation between D2D nodes.
\end{abstract}


\section{Introduction}
\IEEEPARstart{T}{he} homogeneous poisson point process (PPP) is characterized with remarkable property of complete spatial randomness. This property is useful when underlying points are completely uncorrelated with each other and, subsequently, distributed homogeneously. For example, call arrival in cellular networks can precisely be modeled by temporal PPP if we ignore traffic inhomogeneity during day and night times. The spatial version of PPP (SPPP) is extensively used to model position of base stations (BSs) and mobile terminals (MTs) \cite{6515339, 6171996, Lee}. However, neither BSs/MTs are uncorrelated nor distributed homogeneously. Moreover, due to spatial variations in traffic, the intensity measure of the point process cannot be considered constant. The inhomogeneity and spatial correlation is usually governed by several dominant factors such as strategical deployments of BSs, homophilic relations between MTs, emergence of mobile social networks, and existence of hot-spots. As a result, homogeneous PPP is too conservative to model temporal/spatial topologies of network entities. Such point process cannot precisely model cellular networks since it cannot capture negative correlation, in case of BSs, and positive correlation, in case of MTs. The relevant processes that capture negative and positive spatial correlations are fermion and boson \cite{Shirai2003414}. These processes can, respectively, be modeled by Determinantal Point Process (DPP) and Permanental Cox Process (PCP) \cite{j_hough_determinantal_2006}. The PCP is a doubly stochastic Poisson process with intensity measure governed by chi-square random field ($\chi^2_k$-RF) with $k$ degrees of freedom (df). 

\subsection{Related Work}
The negative correlation between BSs are modeled using DPP and Ginibre point process \cite{7155510,6841639,7037373}. However, to the best knowledge of authors, the spatial modeling of MTs is restricted to homogeneous SPPP in the literature. This paper is the first attempt to model inhomogeneous  distribution of MTs with spatial correlation that exists due to any homophilic relation. In this paper, we extend our SPPP approach \cite{7164238,7145780} to PCP model with random intensity measure and inhomogeneous  distribution to characterize interference in underlay D2D network. To validate simulated realizations, we used $n$th-order product density of PCP to derive Ripley's $K$ and variance stabilized $L$ functions. These functions are compared with benchmark SPPP process to see the deviations. The more upper deviations mean high positive correlation between points of the process. The $K$ and $L$ functions of various point processes are available \cite{lavancier2015determinantal,6841639}, however no analytic expressions for PCP exist in the literature. The random intensity measure of PCP is approximated by topological inference based on expected Euler Characteristic (EC) \cite{random_2007}. This approximation is used to derive nearest neighbor distribution function which is introduced into Laplace functional of SPPP to capture interference due to D2D pairs. We propose $u$ and $r$ as interference control parameters to analyze and ensure coverage probability of cellular user.
\subsection{Contributions}
\begin{enumerate}
\item Using $n$th-order product density of PCP, we derive $K$ and $L$ functions for exponential covariance function.
\item We propose expected EC based novel approach to approximate random intensity measure of PCP which is governed by $\chi^2_k$-RF with $k$ df.
\item Inspired by statistical parameter mapping (SPM) and random field theory (RFT) approaches towards functional analysis of brain imaging \cite{friston_statistical_2007}\footnote{In standard functional analysis of brain and neuroimaging, two approaches are followed to identify activation regions against the null hypothesis (e.g., z-test, $\chi^2$-test, t-test, F-test), (i) Bonferroni correction, (ii) Random Field Theory. The functional analysis of brain comprises large number of voxels i.e., large number of statistic values. In case, the statistic values are completely independent, the former approach is best to identify activation regions. However, in multiple comparison problem, spatial correlation always exist and hence later approach provides less conservative analysis and accurate identification of activation regions.}, we adopted RFT approach to derive closed-form approximation (CFA) for intractable nearest neighbor distribution function $G$.
\item We introduce $G$ function into Laplace functional of SPPP to capture interference and subsequently derive CFA for average coverage probability of cellular user in D2D underlay network.
\item We propose $u$ and $r$ as interference control parameters to characterize intracell interference and analyze coverage probability of cellular user in underlay D2D network.
\end{enumerate}
\subsection{Mathematical Preliminaries}\label{math}
\subsubsection{Permanental Cox Process}
We define spatial point process $\Phi$ in terms of $n$th-order product density $\varrho^{(n)}$. The $\Phi$ is a random subset $X$ of underlying locally compact topological/parameter space $S$, a subspace of stratified manifold $M\subset \R^k$. The $\Phi$ is said to be PCP process if $X$ is poisson process with random intensity measure defined as \cite{moller_log_1998}:
\begin{align}
\Lambda(B) \overset{a.s.}=& \int_B \lambda(s) ds,	\IEEEnonumber
\\ =& \int_B \big[Y_1^2(s) + \cdots + Y_k^2(s)\big] ds,	\IEEEnonumber
\\ =& \int_B \chi^2_k ds,
\label{lamphip}
\end{align}
where $B \subseteq S$ is a Borel set, $\lambda(s)$ is random intensity function and $Y_{(\cdot)}(s)$ are $k$ independent Gaussian Random Fields (GRFs).
\subsubsection{Random Field (RF)}
An RF $f = f(t)$ on $M$ can be defined as a function whose values are random variables (RVs) for any $t \in M$ \cite{abrahamsen1997review}. This function is fully characterized by its finite-dimensional distributions (fidi) i.e.,
\begin{align}
F_{t_1,...t_k}(x_1,...,x_k) =& p\big(f(t_1) \le x_1,...,f(t_k) \le x_k\big).
\label{fidi}
\end{align}
In case, (\ref{fidi}) is multivariate Gaussian, $f$ is known as GRF. In real world, not all RFs are Gaussian. Non-Gaussian fields form very broader class and are not well defined. Here, we will consider RFs of the form $g(t) = G(f_m(t)) = G(f_1(t), . . . , f_k(t ))$ as non-Gaussian or Gaussian related RFs. In case $f_1(t),...,f_k(t)$ are zero mean and unit variance GRFs, we can define $\chi^2_k$-RF as \cite{worsley1994local}:
\begin{align}
g(t) =& \sum_{m = 1}^{k} f^2_m(t).
\label{chirf}
\end{align}
The marginal distribution of (\ref{chirf}) for each $t \in M$ is $\chi^2$ with $k$ degrees of freedom.
\subsubsection{Excursion Set}
The excursion set, $A_u$ above level $u \in \R$, of $k$-dimensional RF on $M$ is given as \cite{adler1976,3212831}:
\begin{align}
	A_u(f, M) \triangleq [t \in M: f(t) \ge u] \equiv f^{-1}([u, +\infty)).
\end{align}
The excursion set of a real-valued non-Gaussian RF can be defined by applying function composition $g = (G \circ f)$ on $M$. This set is equivalent to the excursion set of vector-valued Gaussian $f$ in $G^{-1}[u,+\infty)$, which, under appropriate assumptions on $G$, is a manifold with piece-wise smooth boundary given by \cite{random_2007}:
\begin{align}
A_u\big(g,M\big) =& A_u\big((G\circ f), M\big),	\IEEEnonumber
\\	     =& \{t \in M:(G \circ f)(t) \ge u\},		\IEEEnonumber
\\	    =&  \{t \in M: f(t) \in G^{-1}[u,+\infty)\},	\IEEEnonumber
\\	    =& M \cap f^{-1}(G^{-1}[u, +\infty)).
\label{excursionset}
\end{align}
Since, $G^{-1}[u, +\infty)$ is a specific stratified manifold in $\R^k$, we can generalize it to $D \subset \R^k$ in (\ref{excursionset})
\begin{align}
A_u(g,M) =& M \cap f^{-1}(D).
\label{excursionset1}
\end{align}
\subsubsection{Lipschitz-Killing Curvature Measures}
The Euler Characteristic $\mathcal{X}$ is a fundamental additive functional that counts topological components of $M$. In order to consider boundaries, curvatures, surface area, and volume of $M$, the position and rotation-invariant generalized functionals are considered which are known as Lipschitz-killing curvature measures. They are also known as geometric identifiers that capture intrinsic volume of $M$. For example, in case of $M \subset \R^2$, $\mathcal{L}_0 \equiv \mathcal{X}$, $\mathcal{L}_1$, $\mathcal{L}_2$ gives EC, boundary length, and area of manifold $M$. The Lipschitz-killing curvature measures $\mathcal{L}_j$, on $B_R^N$, $N$-dimensional ball of radius $R$, is given as \cite[Section 6.3]{random_2007}:
\begin{align}
\mathcal{L}_j (B_R^N)&= 
\begin{pmatrix}
           	N \\
           	j
         \end{pmatrix}
	R^j \frac{w_N}{w_{N-j}},
\label{intvol}
\end{align}
where $j$ has dimension $M$ (i.e., $j = |M|$) and $\omega_n$ is the volume of the unit ball in $\mathbb{R}^n$.
\subsection{Organization}
The rest of the paper is organized as follows: In Section II, we present system model of cellular network with underlay D2D communication. This is followed by PCP model, for spatial distribution of potential D2D nodes, and details to generate such a process. To validate the simulated realizations of PCP, we also derive CFA of $K$ and $L$ functions in this section. In Section III, we present the main result of approximating $G$ function of PCP. The CFA of $G$ function is derived based on expected EC of excursion set of $\chi^2_k$-RF. The CFA and empirical $G$ function are compared with SPPP. In Section IV, we introduce $G$ function into Laplace functional of SPPP to derive CFA for average coverage probability of cellular user. Conclusion is drawn in Section V.
\section{System Model}\label{sysmod}
In this section, we present cellular network model, PCP model, process generation, and validation using $K$ and $L$ functions.
\subsection{Cellular Network Model}
The cellular network comprises small cell BS (SBS) and MTs as shown in Fig. \ref{Figure:sm}. We consider orthogonal frequency division multiple access (OFDMA) based cellular network. In this network, we want to analyze maximum frequency reuse and the effect of interference due to D2D communication. In order to avoid coverage holes, SBS should provide homogeneous coverage for the cellular user. For the homogeneous coverage at each position in the coverage area, we consider one MT as cellular user which is distributed uniformly. All other MTs are potential D2D users distributed according to PCP process. 

The uplink resources of cellular user are shared by potential D2D users. The time division duplex (TDD) mode is assumed between D2D nodes to capture the effect of interference by both nodes. In case of frequency division duplex (FDD) mode, the interference at any time instant will simply be half that of TDD mode. The data and signaling is provided by the serving SBS to the cellular user whereas only signaling is assumed for potential D2D nodes. For average coverage probability of cellular user, interference is generated by all successful D2D pairs. We consider negligible interference at serving SBS from successful D2D nodes in neighboring SBSs due to negligibly small transmit power.
 
The cellular and potential D2D users are distributed in the coverage area bounded between SBS radius $R$ and the protection region $R_0$. The distance between SBS and cellular user is $r_c$ which is used to calculate path-loss at SBS. Every successful D2D pair has a distance of $r$ between nodes. The channel model assumes distance dependent path-loss and Rayleigh fading. The simple singular path-loss model $({r_c}^{-\alpha})$ is assumed where the protection region $R_0$ ensures the convergence of the model by avoiding nodes to lie at the origin. The received power at SBS follows exponential distribution. The distance $r_c$ follows uniform distribution \cite{7164238}:
\begin{align}
f(r_c)=\frac{2r_c}{R^2},f({\theta})=\frac{1}{2\pi},
\label{frc}
\end{align}
where $R_0\leq r_c \leq R$ and $0 \leq \theta \leq 2\pi$.
\begin{figure}[t]
\centering
\includegraphics[width = 1\columnwidth]{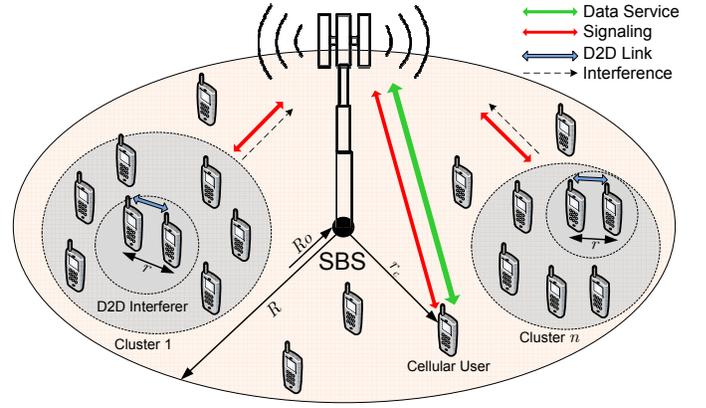}
\caption{\label{network} System model for PCP deployed MTs.\label{Figure:sm}
\vspace{-1mm}}
\end{figure}
\begin{figure*}[t]
\centering
    \subfigure[Intensity function of PCP resembles bell-like blobs of GRF.]
    {
        \includegraphics[width = 3.25 in]{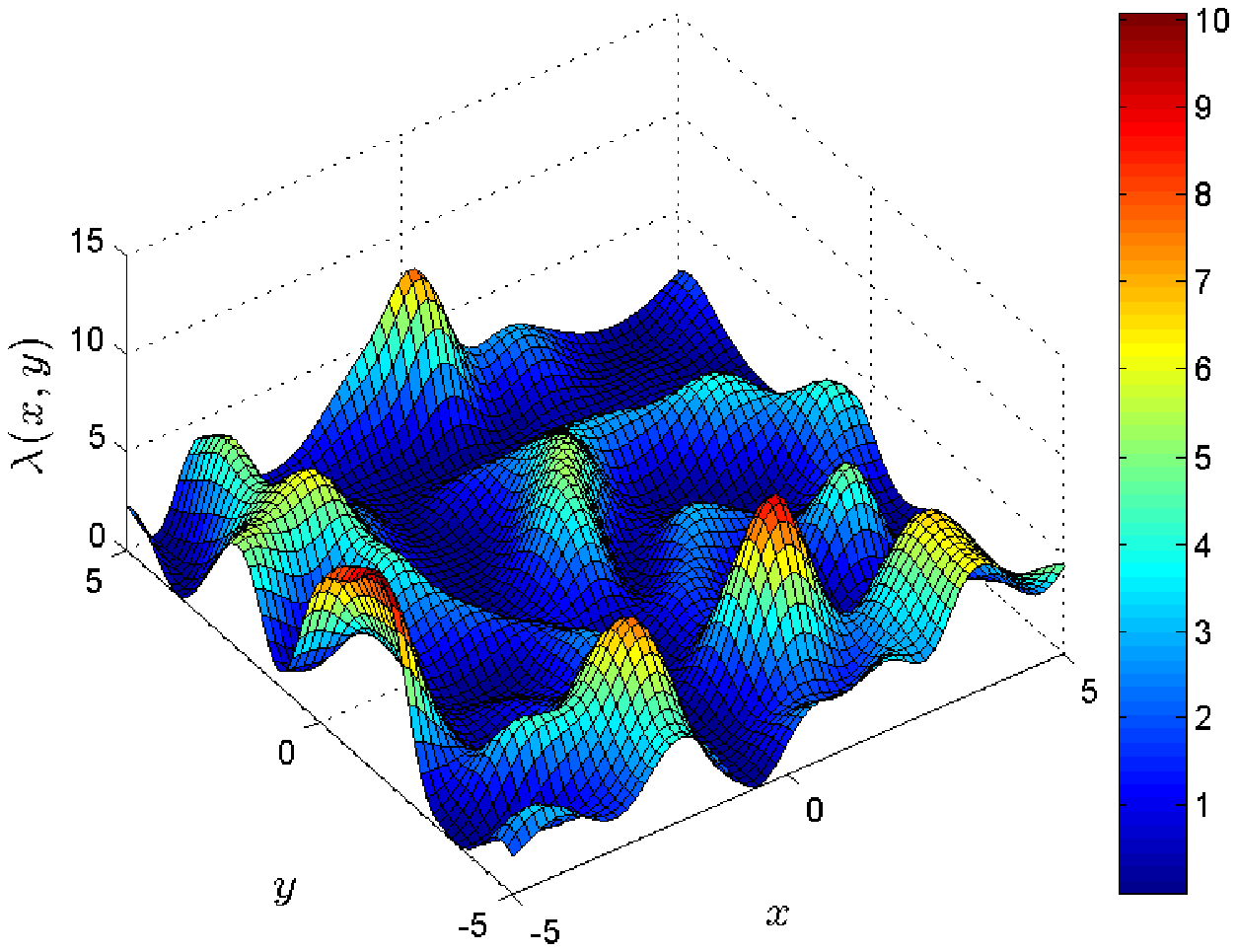}
        \label{Figure:fieldpcp}
    }
      \subfigure[Sampled realization of PCP shows high and low intensity regions.]
    {
        \includegraphics[width = 3.25 in]{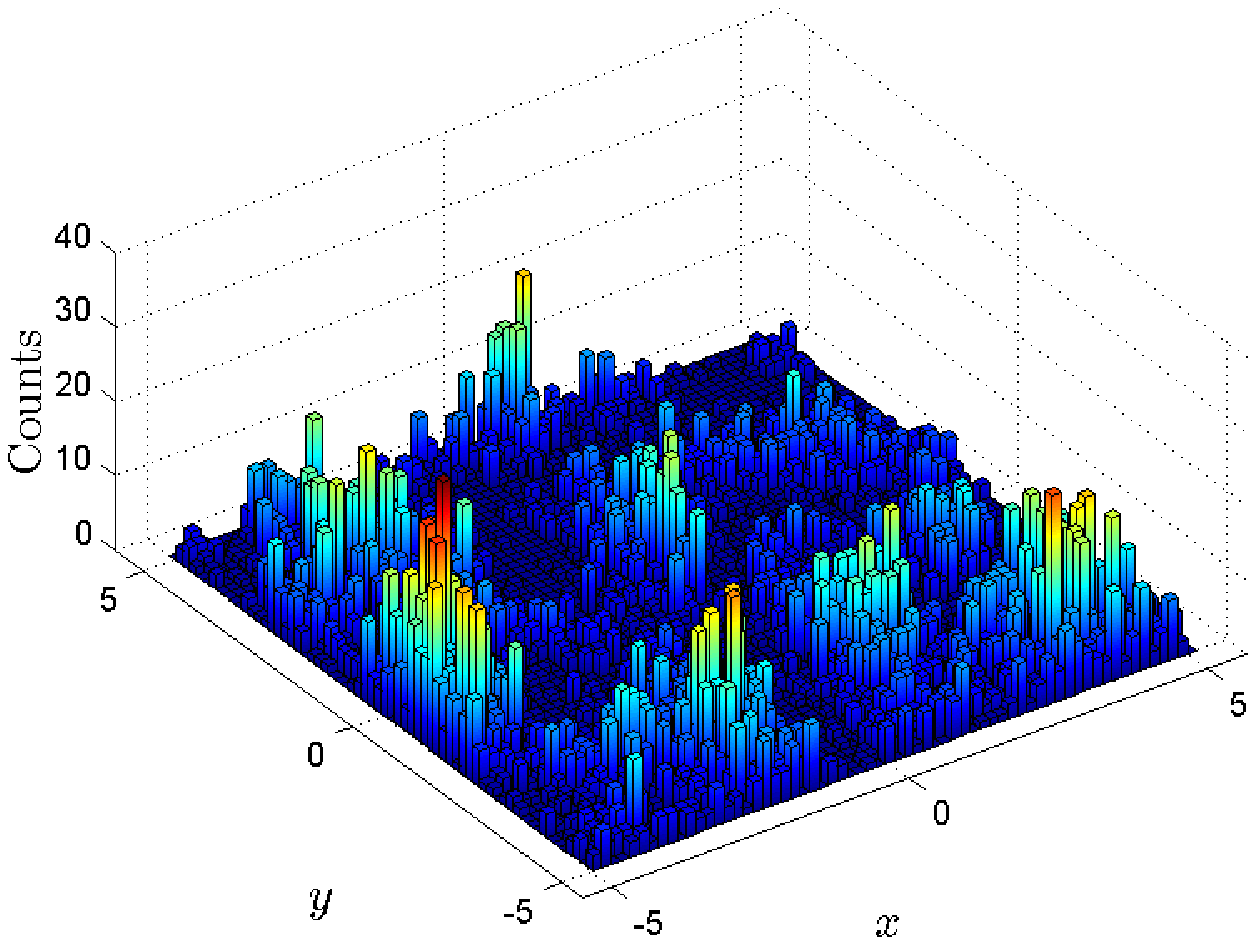}
        \label{Figure:hist3pcp}
    }
 \caption{Random intensity function of PCP process and sampled histogram.}
\label{Figure:fieldhist3}
\end{figure*}
\subsection{PCP Model}\label{pcpmodel}
The $n$th-order product density $\varrho^{(n)}$ of a Cox process is \cite{moller_2006}:
\begin{align}
\varrho^{(n)} (s_1,...,s_n) = \mathrm{E} \prod_{i = 1}^{n} \Lambda(s_i),
\end{align}
where $\Lambda(\cdot)$ is a random intensity measure. In order to model spatially correlated process for potential D2D nodes, we consider PCP with the following intensity measure:
\begin{align}
\Lambda(s_i) = Y_1^2(s_i) + \cdots + Y_k^2(s_i),
\label{impcp}
\end{align}
where $Y_{(\cdot)}(\cdot)$ are zero mean unit variance $k$ independent real-valued stationary GRFs and $Y_{(\cdot)}^2(\cdot) \sim \chi^2(\cdot)$ with unit df.

The sum of independent chi-square distributions has remarkable property given by the following theorem \cite{ross_introduction_2009}:
\begin{theorem}
The sum of $k$ independent chi-square distributions with $v_i$ df follows a chi-square distribution with $\sum_{i=1}^k v_i$ df i.e.,
\begin{align}
Z =& \chi^2_{v_1}(s_i)+...+\chi^2_{v_k}(s_i),	\IEEEnonumber
\\ & \sim \chi^2_{v_1+...+v_k}(s_i).
\label{z}
\end{align}
\end{theorem}
Each squared GRF in (\ref{impcp}) has unit df ($v_i =1$); this results in $Y_{(\cdot)}^2(s_i) \sim \chi^2_k(s_i)$. Therefore, the intensity measure of PCP is governed by $\chi^2_k$-RF with $k$ df as:
\begin{align}
\Lambda(s_i)  =&  \chi^2_k(s_i).
\label{impcp1}
\end{align}
Since the distribution of potential D2D nodes is translation and motion invariant, we can assume stationary PCP and hence borrow the definition from \cite{moller_log_1998}:
\begin{defn}
A Permanental Cox Process is stationary if and only if the underlying GRF is stationary.
\end{defn}

The stationarity of GRF is ensured by underlying covariance function. In order to generate smooth GRFs, we consider squared exponential covariance function
\cite{rasmussen_gaussian_2006}
\begin{align}
	C(s_1, s_2) =  e^{-\frac{||s||^2}{2 l^2}},
	\label{cov}
\end{align}
where $||s|| = ||s_1-s_2|| = \sqrt{(s_{1_x}-s_{2_x})^2+(s_{1_y}-s_{2_y})^2}$ is the Euclidean distance between $s_1$ and $s_2$, and $l$ is the characteristic length-scale. The resulting covariance matrix of PCP is represented by $[C] (s_1,..., s_n)$.

The fact that PCP is a type of Permanental point process is due to the following theorem \cite{moller_2006}:
\begin{theorem}
The $n$th-order product density of Cox process is equal to the weighted permanent of the covariance matrix i.e.,
\begin{align}
\varrho^{(n)} (s_1,...,s_n) = \textnormal{per}_{\alpha} [C] (s_1,..., s_n).
\label{pd}
\end{align}
\end{theorem}
\begin{proof}
See, \cite[Sec. 2.1.1, p. 876, Theorem 1]{moller_2006}.
\end{proof}
The boson (or photon) process corresponds to $\alpha$ = 1 \cite{Macchi1975} resulting in $k = 2\alpha$ df for underlying GRFs of PCP.
\subsection{Process Generation}\label{pg}
The random intensity measure of PCP is governed by $\chi^2_k$-RF which is non-Gaussian or Gaussian related RF. This field is generated by squaring the component field which is GRF (as discussed in Section \ref{pcpmodel}). The GRF is a collection of RVs with fidi as multivariate Gaussian. Therefore, it can be generated by drawing real valued multivariate normal random vectors and mapping it to the underlying grid. It can also be generated via circulant embedding method \cite{spodarev_stochastic_2013}. We followed the former approach to generate RFs and subsequently PCP process. The $\chi^2_k$-RF of PCP comprises large number of $\chi^2$ RVs which are mapped to each grid point $s \in S$. Due to smooth underlying covariance structure, each RV results into smooth sample path. The blobs and holes show spatial covariance between $\chi^2$ RVs. The overall shape of intensity measure of PCP is similar to symmetric bell-like blobs of GRF, however, the loss of symmetry in this case is due to low df of $\chi^2_k$-RF. For large df, due to central limit theorem, the intensity measure of PCP resembles symmetric bell-like blobs of GRF.

The lattice representation of $\chi^2_k$-RF with 2 df is shown in Fig. \ref{Figure:fieldpcp}. In this figure, we can see number of blobs and holes which, respectively, show high and low intensity areas. The high intensity areas (7 $\rightarrow$ 10 on colorbar) capture strong spatial correlation between points and results in group clustering whereas low intensity areas (1 $\rightarrow$ 3 on colorbar) form holes due to nonexistence of any homophilic relation. The Markov chain Monte Carlo based Metropolis-Hasting (MH) sampler\footnote{The MH sampler is used to sample RVs from multidimensional spaces. The states of the underlying Markov chain can be updated in two different ways, (i) Block-wise, (ii) Component-wise. The first approach updates all state variables simultaneously whereas the second approach iterates with component-wise update. In both the cases, the acceptance probability is $\alpha = \textnormal{min}\bigg(1, \frac{p(\pmb \theta^*)}{p(\pmb \theta^{(t-1)})} \frac{q(\pmb \theta^{(t-1)}|\pmb \theta^*)}{q(\pmb \theta^*|\pmb \theta^{(t-1)})}\bigg)$, where $p(\pmb \theta)$ and $q(\pmb \theta)$ stand for proposal and target distributions, respectively \cite{martinez_computational_2007}.} is used to sample PCP points under $\chi^2_k$-RF as shown in Fig. \ref{Figure:hist3pcp} which shows inhomogeneous and clustered distribution of points.
\subsection{Summary Statistics: $K$ and $L$ Functions}
The Ripley's $K$ function and variance stabilized $L$ functions are, respectively, given as \cite{Ripley}:
\begin{align}
	K(r) =& \int_0^{r} g(s) 2 \pi s ds,
	\label{Kfunc}
\\	L(r) =& \sqrt{\frac{K(r)}{\pi}},
	\label{Lfunc}
\end{align}
where $r$ is the distance and $g(s)$ is the pair correlation function.

Using (\ref{pd}), the first and second order product densities can be derived as \cite{mccullagh2005permanent}:
\begin{align}
	\varrho = \alpha C(0),	\quad\quad	\varrho^{(2)}(s) = \alpha \big[1 + C^2(s)\big].
	\label{rhos}
\end{align}
Since $g(s) = \varrho^{(2)}(s)/\varrho^2$, the pair correlation function of PCP for $\alpha$ = 1 is given by:
\begin{align}
	g(s) =& 1 + C^2(s).
	\label{gr}
\end{align}
The corresponding $K$ and $L$ functions can be derived as:
\begin{align}
	K(r) =& \pi r^2 + \pi l^2(1 - e^{-(\frac{r}{l})^2}),	\IEEEnonumber
	\\
	L(r) =& \sqrt{r^2 + l^2(1 - e^{-(\frac{r}{l})^2})}. 		\IEEEnonumber
\end{align}

The SPPP is a special case of PCP with $g(s) = 1$, $K(r) = \pi r^2$, and $L(r) = r$.
If we assume complete spatial independence, the covariance between $s_1$ and $s_2$ vanishes and product densities from (\ref{rhos}) reduces to:
\begin{align}
	\varrho = \alpha,	\quad\quad	\varrho^{(2)}(s) = \alpha \big[1 + 0\big].	\IEEEnonumber
\end{align}
The corresponding pair correlation function (for $k = 2$ i.e., $\alpha = 1$) is 1. The $K$ and $L$ functions for SPPP can be validated as $K(r) = \pi r^2$ and $L(r) = r$, respectively.
\begin{figure}[t]
\centering
\includegraphics[width=1\columnwidth]{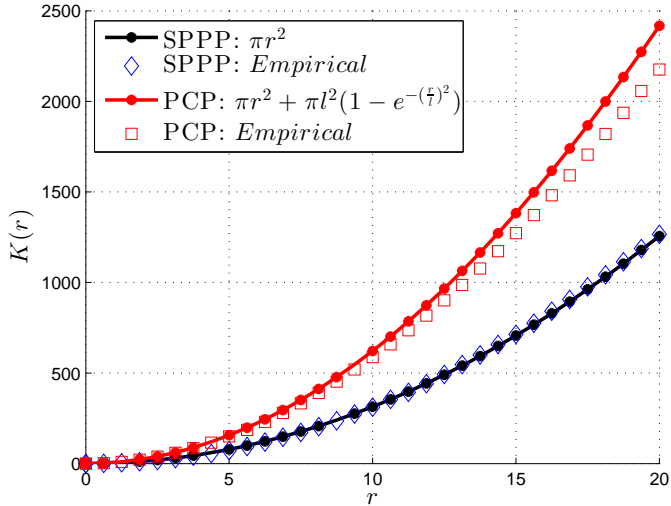}
\caption{CFA and empirical approximations for $K$ function of SPPP and PCP.}\label{Figure:kfunc}
\end{figure}

The estimated $\hat{K}$ function for SPPP and PCP is, respectively, given as \cite{marcon2003generalizing}:
\begin{align}
	\hat{K}_{SPPP}(r) =& \frac{1}{\lambda n} \sum_{i=1}^{n} \sum_{j=1, i\neq j}^{n} I(s_i,s_j,r),	\IEEEnonumber
\end{align}
\begin{align}
	\hat{K}_{PCP}(r) =& \frac{1}{n} \sum_{i=1}^{n} \sum_{j=1, i\neq j}^{n} \frac{I(s_i,s_j,r)}{\lambda(s_i) \lambda(s_j)},	\IEEEnonumber
\end{align}
where $\lambda$ is the constant intensity function of SPPP, $\lambda(\cdot)$ is the random intensity function of PCP and $I(\cdot)$ is the indicator function for the distance $r$ between points $s_i$ and $s_j$. The analytic expression of $K$ and $L$ functions and empirical estimates are shown in Fig. \ref{Figure:kfunc} and Fig. \ref{Figure:lfunc}, respectively. As an illustration, the plot is shown for the value of $l$ = 50. This parameter captures the length-scale of the underlying sample path. In modeling problem, it can be used to incorporate the level of covariance in points of the process. We can see that the estimates of $K$ and $L$ functions matches CFA. In these figures, SPPP serves as a benchmark with zero correlation between points. The positive spatial correlation of PCP can be verified by upper drift of $K$ and $L$ functions. In case of negative spatial correlation, the $K$ and $L$ functions shall lie below SPPP curves as can be seen in \cite{7155510,6841639,7037373}.

\section{Nearest Neighbor Distribution Function}\label{retprob}
In this section, we approximate $G$ function using topological inference based on expected EC and Poisson clumping heuristic \cite{1428384_Cao}.
\subsection{Nearest Neighbor Distribution Function}
The nearest neighbor distribution function of a point process is given as:
\begin{align}
	G(r) =& 1 - \E \big[e^{- \Lambda(B^N_r)}\big],	\IEEEnonumber
	\\    =&  1 - \E \bigg[e^{-\int_{B^N_r} \lambda(s) ds}\bigg],
\label{grpp}
\end{align}
where $\Lambda$ and $\lambda$ are, respectively, intensity measure and function of point process over closed ball of radius $r$ at arbitrary position.

In case of SPPP, $\lambda$ is constant and hence (\ref{grpp}) can be simplified as:
\begin{align}
	G(r) =& 1 - e^{- \lambda \pi r^2}.
\label{grsppp}
\end{align}
\begin{figure}[t]
\centering
\includegraphics[width = 1\columnwidth]{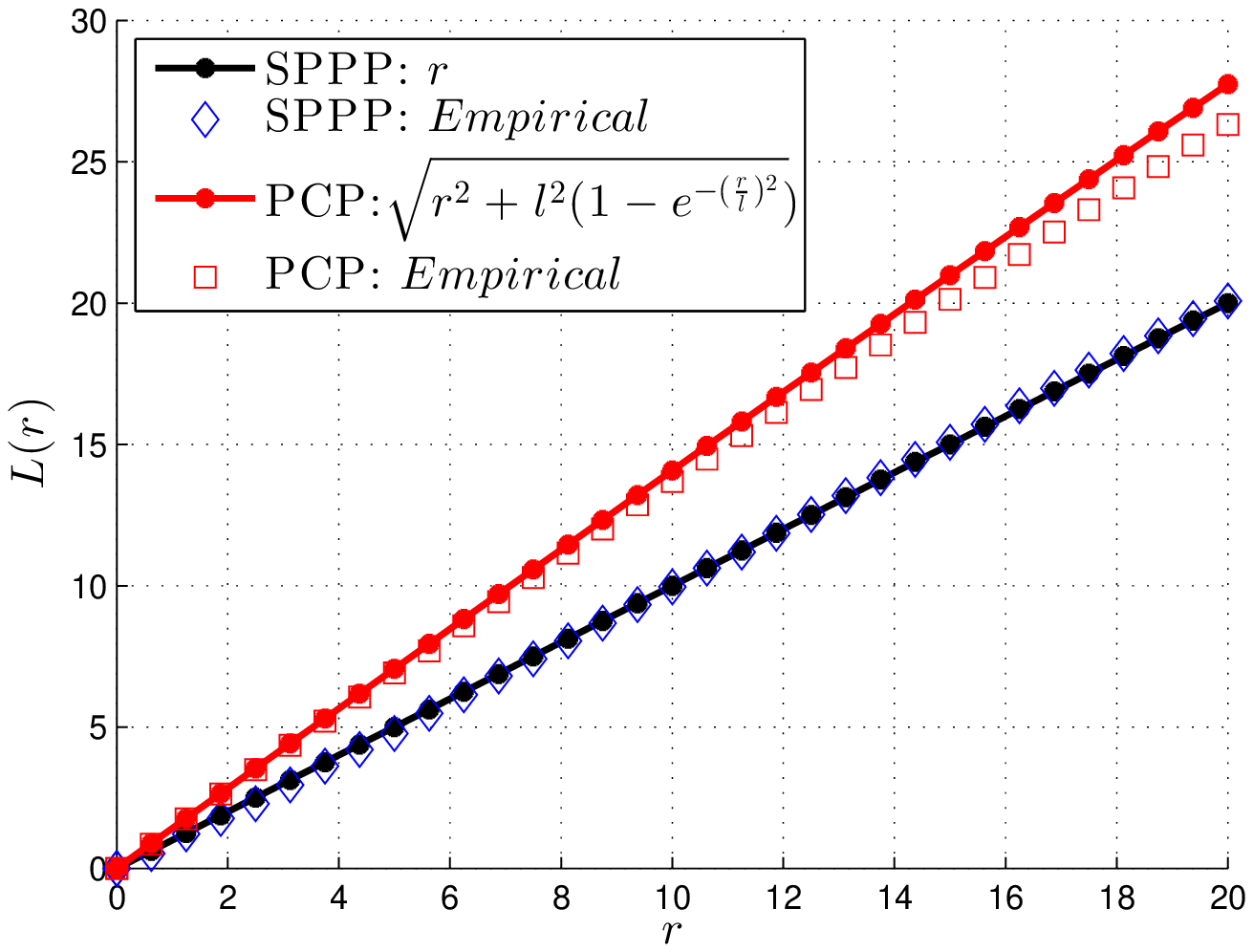}
\caption{CFA and empirical approximations for $L$ function of SPPP and PCP.}\label{Figure:lfunc}
\end{figure}

Considering $\Lambda$ from (\ref{impcp1}), the nearest neighbor distribution function of PCP is given as:
\begin{align}
	G(r) =& 1- \E \bigg[e^{-\int_{B^N_r} \chi^2_k \,ds}\bigg],	\IEEEnonumber
	 \\    =& 1- \E \bigg[e^{-\int_{B^N_r} \int_{v_1} \cdot\cdot\cdot \int_{v_n} \chi^2 dv_1\cdot\cdot\cdot dv_n \, ds}\bigg].
\label{grpcp}
\end{align}
Since $\chi^2_k$-RF is a collection of large number of RVs, this results in evaluation of nested integrals over $B^N_r$ which is mathematically intractable. In this case, we approximate intensity measure $\Lambda$ using expected EC of excursion set of $\chi^2_k$-RF \cite{random_2007}.
\subsubsection{Approximation of Intensity Measure}\label{appim}
The expected intensity measure in (\ref{grpcp}) can be estimated by making topological inference of average number of upcrossings of $\chi^2_k$-RF above level $u$ of excursion set. This approach is based on Gaussian Kinematic Formulae (GKF) given as \cite[Theorem 15.9.4]{random_2007}:
\begin{theorem}\label{th4}
Let $M$ be an $N$-dimensional, regular stratified manifold, $D$ a regular stratified subset of $\R^k$. Let $f = (f_1, ..., f_k) : M \rightarrow \R^k$ be a vector-valued Gaussian field, with independent, identically distributed components and $f$ being Morse function\footnote{For the definition of Morse function and Morse's Theorem, refer \cite[Theorem 4.4.1, pp. 87]{adler_geometry_2009} and \cite[Section 9.3, Definition 9.3.1, pp. 206]{martinez_computational_2007}.} over $M$ with probability one. Then
\begin{align}
\E\big[\mathcal{L}_i(M \cap f^{-1}(D))\big] =& \sum_{j=0}^{N-i}
	\begin{bmatrix}
           	i+j \\
           	j
         \end{bmatrix}
	\frac{\mathcal{L}_{i+j}(M) \mathcal{M}_j(D)}{(2\pi)^{\frac{j}{2}} },
\label{gkf}
\end{align}
where $\mathcal{L}_{i+j}$ for $i = 0, ..., N , j = 0, ..., N-i$, are Lipschitz-Killing curvature measures on $M$ with respect to the metric induced by $f$ and $\mathcal{M}_j$ are the generalized (Gaussian) Minkowski functionals on $\R^k$.
\end{theorem}
For notational convenience, we assume the combinatorial flag coefficients $\begin{bmatrix}
           	i+j \\
           	j
         \end{bmatrix} = \begin{bmatrix}
           	a \\
           	b
         \end{bmatrix}$ given as:
\begin{align}
	\begin{bmatrix}
		a\\
		b
	\end{bmatrix} = \frac{[a]!}{[b]! [a-b]!}, \quad [a]! = a! \omega_{a},\quad \omega_{a} = \frac{\pi^{\frac{a}{2}}}{\Gamma(\frac{a}{2}+1)}.	\IEEEnonumber
\end{align}

Using Theorem \ref{th4} and putting $M \cap f^{-1}(D)$ from (\ref{excursionset1}), the expected intensity measure can be approximated as follows:
\begin{align}
\psi_0 \approx & \E \big[\chi^2_k(B^N_r)\big],	\IEEEnonumber
\\ =& \E \big[\mathcal{L}_{0}\big(A_u(\chi^2_k, B^N_r)\big)\big],	\IEEEnonumber
\\ =& \sum_{j=0}^{N} \frac{\mathcal{L}_{j}(B^N_r) \mathcal{M}_j(D)}{(2\pi)^{\frac{j}{2}} }.
\label{Elambda}
\end{align}
The Minkowski functionals $\mathcal{M}_j(D)$ can be transformed into EC density for $\chi^2_k$-RF as \cite{random_2007}:
\begin{align}
\psi_0 =& \sum_{j=0}^{N} \rho_j(u) \mathcal{L}_{j}(B^N_r),	
\label{Elambda1}
\end{align}
where
\begin{align}
\rho_j(u) =& \frac{u^{\frac{k-j}{2}}e^{-\frac{u}{2}}}{(2\pi)^{\frac{j}{2}}\Gamma(\frac{k}{2})2^{\frac{k-2}{2}}} \sum_{l=0}^{\floor{\frac{j-1}{2}}} \sum_{m=0}^{j-1-2l}	\IEEEnonumber
\\	& \times \mathbbm{1}_{\{k\ge j-m-2l\}} 	\begin{pmatrix}
		k-1\\
		j-1-m-2l
	\end{pmatrix}	\IEEEnonumber
\\ 	& \times \frac{(-1)^{j-1+m+l} (j-1)!}{m!l!2^l}\:u^{m+l}.\IEEEnonumber
\end{align}

The EC density over $\R^2$ and average upcrossings of $\chi^2_k$-RF are shown in Fig. \ref{Figure:ecdensity} and Fig. \ref{Figure:avgupcross}, respectively.
\begin{figure}[t]
\centering
\includegraphics[width=1\columnwidth]{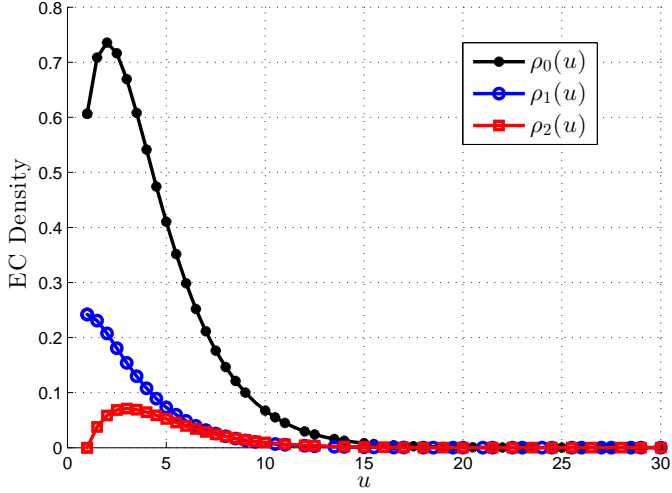}
\caption{EC density representing density of unit component of $\chi^2_k$-RF that survives threshold $u$.}\label{Figure:ecdensity}
\end{figure}
\begin{figure}[t]
\centering
\includegraphics[width = 1\columnwidth]{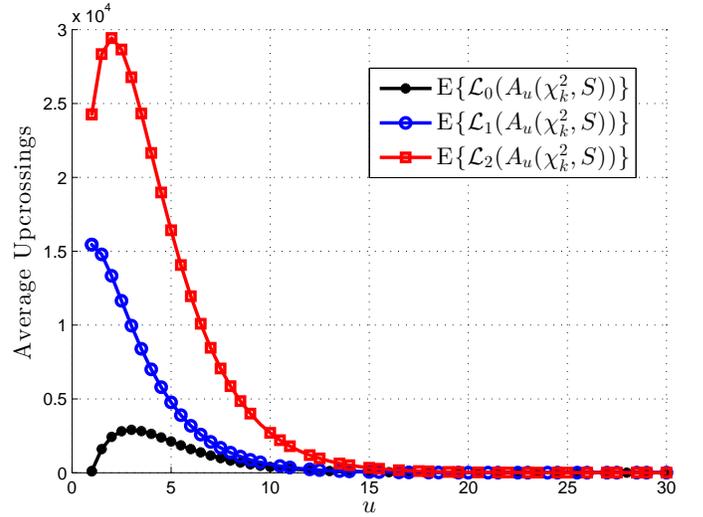}
\caption{Average number of components of $\chi^2_k$-RF for different Lipschitz-killing curvature measures over 200 $\times$ 200 square grid.}\label{Figure:avgupcross}
\end{figure}
In Fig. \ref{Figure:ecdensity}, a very interesting fact can be observed for $\rho_0$. The value of $j = 0$ transforms $N$ dimensional Manifold $M$ to a single point where spatial correlation does not make sense. In this case, EC density reduces to $\chi^2$ distribution with $k$ df conforming the fact that marginal distribution of $\chi^2_k$-RF is $\chi^2$ distributed \cite{worsley1994local}. In Fig. \ref{Figure:avgupcross}, the behavior of $\chi^2_k$-RF for different $u$ is plotted with respect to different $\mathcal{L}_{i}$ measures. In this paper, we consider $\mathcal{L}_0$ (i.e., EC) to approximate intensity measure of PCP process.
\subsubsection{Poisson Clumping Heuristic\protect \footnote{At high thresholds, the clusters in the excursion set can be regarded as multidimensional point process with no memory and hence they behave as poisson clumps \cite{aldous_probability_2013}.}}
To approximate nearest neighbor distribution function, we consider probability of getting one, or more, clusters (D2D pairs) with spatial extent $r$, or more, above threshold $u$. The general expression for this cluster level inference is given as \cite{friston_statistical_2007,mazziotta_brain_2000}:
\begin{align}
G(r) \simeq 1 - e^{- \psi_0 \, p(v \ge r)}.
\label{pumc}
\end{align}
The volume $v$ of clusters (D2D pairs) over spatial extent $r$ is distributed according to \cite{friston_assessing_1994}:
\begin{align}
p(v \ge r) \approx e^{\big(-\beta r^{\frac{2}{N}}\big)},
\label{pvgek}
\end{align}
where
\begin{align}
\beta = \bigg(\frac{\Gamma(\frac{N}{2}+1)}{\eta}\bigg)^{\frac{2}{N}},	\IEEEnonumber
\end{align}
and $\eta = {\rho_0}^v \mathcal{L}_N/\psi_0$ is the expected volume of each cluster. 

The plots of $\eta$ and $p(v\ge r)$ are shown in Fig. \ref{Figure:eeta} and Fig. \ref{Figure:probvgek}, respectively.
\begin{figure}[t]
\centering
\includegraphics[width = 1\columnwidth]{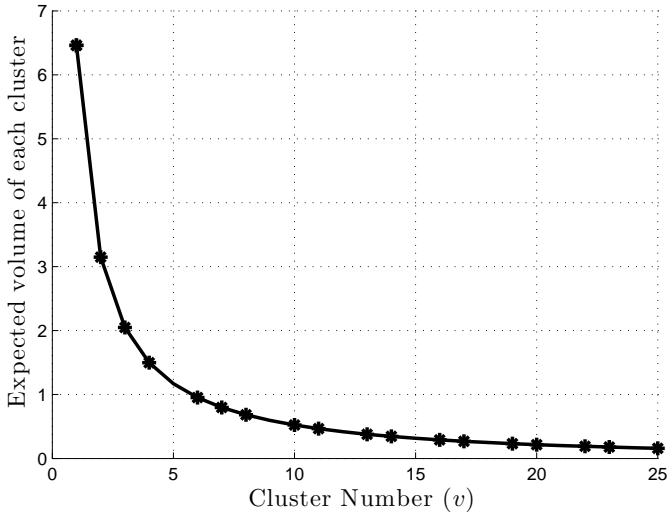}
\caption{Expected volume of each cluster depends on EC density, volume of underlying space, and expected EC. The plot is drawn for $u$ = 31.}\label{Figure:eeta}
\end{figure}
In these figures, we can see that the maximum expected volume and the probability to have nodes with spatial extent ($\ge) r$ occurs for $v$ = 1. In deriving CFA of $G$ function, we, consider $\eta$ for $v$ = 1.

Using (\ref{intvol}), the Lipschitz-Killing curvature measures $\mathcal{L}_j$ over ball of spatial extent $r$ can be derived as:
\begin{subequations}
\begin{align}
\mathcal{L}_0 &= 1	,
\\ \mathcal{L}_1 &= 2 \sqrt{\pi} r \frac{\Gamma(\frac{1}{2}+1)}{\Gamma(2)},
\\ \mathcal{L}_2 &= \pi r^2 \frac{\Gamma(1)}{\Gamma(2)}.
\end{align}
\label{intvol1}
\end{subequations}
Considering (\ref{intvol1}), (\ref{Elambda1}), and (\ref{pvgek}) in (\ref{pumc}) for, at most, distance r, the $G$
function can be approximated. The plot of $G$ function can be seen in Fig. \ref{Figure:gfunc} where PCP points, due to positive spatial correlation, have higher probability of D2D pairs as compared to SPPP points. For example, at a distance of 2.5 m, the probability of two spatially correlated potential candidates for D2D communication is 0.8 as compared to 0.42 in SPPP points which occur so close by chance (i.e., not due to spatial correlation under some homophilic relation). This is because SPPP cannot model spatial correlation between points and is characterized by complete spatial randomness.
\section{Average Coverage Probability}\label{covprob}
In this section, we introduce $G$ function as retention probability of D2D nodes at spatial extent $r$ to analyze the interference and resulting average coverage probability of cellular user. We assume interference-limited environment due to large number of potential D2D pairs. Hence, the signal-to-interference ratio (SIR) is given as:
\begin{align}
\textnormal{SIR}_{SBS} =  \frac{p_c f_{c} r^{-\alpha}_{c}}{\sum_{i \in \Phi} p_{i} f_{i} r^{-\alpha}_{i}},
\label{SIR_S1}
\end{align}
where $p_{c}$ and $p_{i}$ are transmit powers of cellular user and D2D interferers, respectively; $f_{c}$, and $f_{i}$ are respective small-scale fading. The corresponding distance dependent path-loss are $r^{-\alpha}_{c}$ and $r^{-\alpha}_{i}$.
\begin{figure}[t]
\centering
\includegraphics[width = 1\columnwidth]{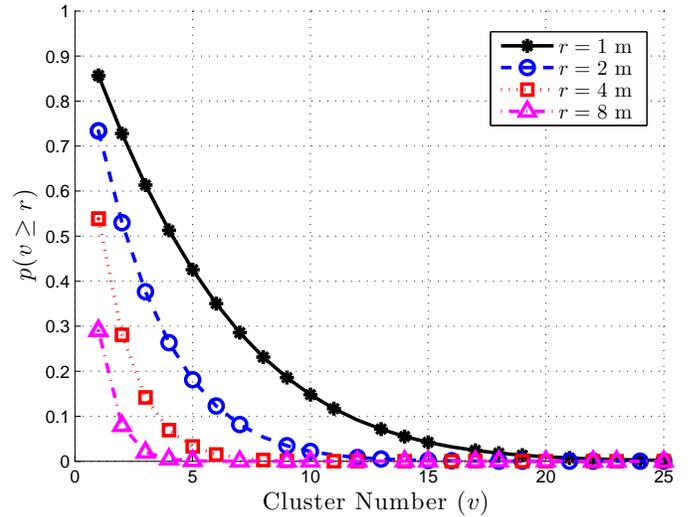}
\caption{The probability of cluster $v$ with different spatial extents $r$ is exponentially distributed with mean $\beta$ \big(see Eq. (\ref{pvgek})\big).}\label{Figure:probvgek}
\end{figure}
\begin{figure}[b]
\centering
\includegraphics[width=1\columnwidth]{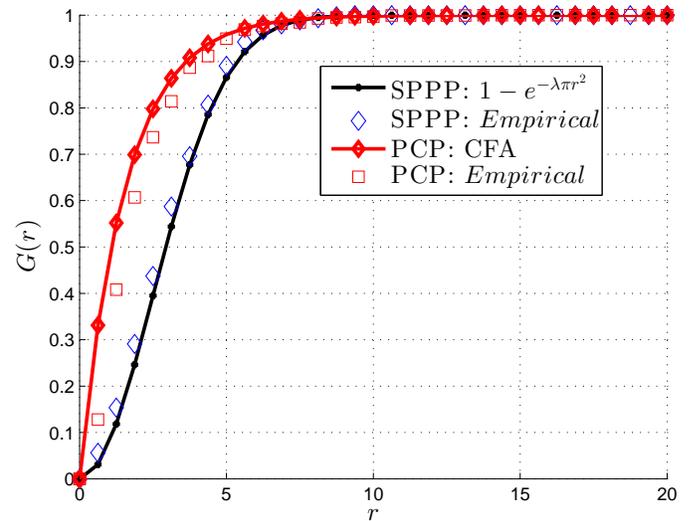}
\caption{Nearest neighbor distribution function of PCP shows high probability for lower values of $r$ as compared to homogeneous SPPP process.}\label{Figure:gfunc}
\end{figure}
Assuming exponential distribution with mean 1 for power received by SBS, the average coverage probability of uplink cellular user is given by the following theorem.
\begin{theorem}
The average coverage probability of a cellular user with underlay D2D communication is
\begin{align}
p_{cov}^{c} = & \int^{R}_{R0}e^{-\frac{2\pi^2 \psi_0 p(r)\,r^{2}_{c}}{\alpha \sin(\frac{2 \pi}{\alpha})} \big(\frac{\gamma}{ p_{c}}\big)^{\frac{2}{\alpha}} \mathbb{E}[{p^{\frac{2}{\alpha}}_{i}}] }\,\frac{2 r_{c}}{R^{2}} dr_{c}.
\label{avgcovprob}
\end{align}
\end{theorem}
\begin{proof}
See Appendix A.
\end{proof}
For same transmit power of all D2D interferers, the average coverage probability of cellular user for path-loss exponent $\alpha = 4$ and $R_0 \sim 0$ reduces to:
\begin{align}
p_{cov}^{c} = & \frac{e^{-\frac{\pi^2 R^2 \psi_0 p(r)}{2} \sqrt{\frac{\gamma p_i}{p_c}}}-1}{-\frac{\pi^2 R^2 \psi_0 p(r)}{2}\sqrt{\frac{\gamma p_i}{p_c}}}.
\label{PGFL_SC3}
\end{align}
\section{Numerical Results And Discussion}
In this section, we numerically evaluate the analytic expressions of Sec. \ref{covprob} by varying the number of different parameters. The average coverage probability of cellular user in (\ref{PGFL_SC3}) depends on $R$, $\psi_0$ (and an important implicit parameter $u$), spatial extent $r$, D2D transmit power $p_i$, transmit power of cellular user $p_c$, and target threshold $\gamma$. The first and foremost step is to identify implicit parameter $u$ which is introduced as interference control parameter for D2D pairing. This parameter is a function of grid size and more specifically SBS radius $R$. By finding the feasible range of $u$ for a given radius $R$, we have varied other parameters to analyze average coverage probability of cellular user. For different spatial extents $r$, the cumulative interference effect is captured. Since the distance between D2D pairs is much smaller as compared to distance between D2D pair and SBS, it is reasonable to assume same transmit power for every D2D pair in the coverage area. To analyze the interference due to D2D clusters, we introduce nearest neighbor distribution function into Laplace functional of SPPP.
\begin{figure*}[t]
\centering
\begin{minipage}{.5\linewidth}
	\centering
	\subfigure[Intensity function of PCP on 200$\times$200 grid  for threshold values $u = (1, 2, 4, 8)$.]{\label{Figure:fieldthresholdfull}\includegraphics[scale=.5]{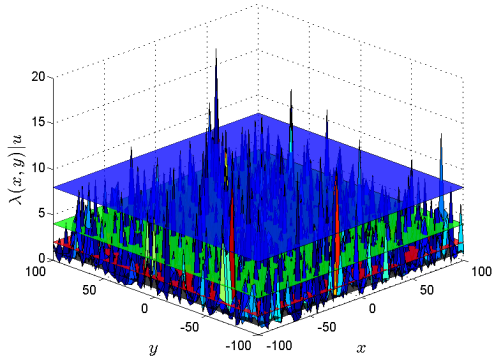}}
	\end{minipage}%
	\begin{minipage}{.5\linewidth}
	\centering
	\subfigure[10$\times$10 extract of Fig. \ref{Figure:fieldthresholdfull} clearly shows surviving and departing blobs at same values of $u$.]{\label{Figure:fieldthreshold}\includegraphics[scale=.5]{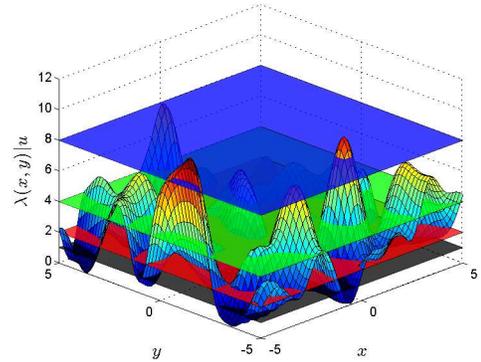}}
	\end{minipage}\par\medskip
\subfigure[Coverage probability of cellular user drops for low values of $u$ i.e., maximum number of possible D2D pairs.]{\label{Figure:covprob}\includegraphics[scale=.5]{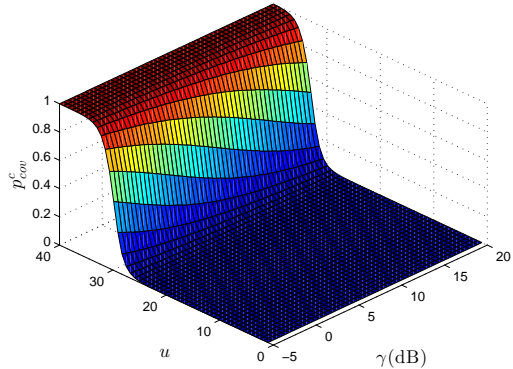}}
\caption{Interference characterization in terms of intensity measure of PCP at different values of $u$ (Fig. \ref{Figure:fieldthresholdfull}, \ref{Figure:fieldthreshold}) and average coverage probability of cellular user (Fig. \ref{Figure:covprob}) for $p_{i} = $ 0 dBm, $p_{c} = $ 20 dBm, and $R =$ 100 m.}
\label{Figure:fieldpcov}
\end{figure*}

In Fig. \ref{Figure:fieldpcov}, we show intensity function at different thresholds and corresponding average coverage probability of cellular user for grid size 200$\times$200 (SBS of radius $R$ = 100 m). For clear illustration, the small portion of this grid has been shown in Fig. \ref{Figure:fieldthreshold}. In this grid, if we set $u = 2$ (transparent black plane), all nodes (despite low spatial correlation) will be considered to make D2D pairs\footnote{The maximum number of upcrossings of $\chi^2_k$-RF occurs at around $u = k$ as can be verified in Euler density (Fig. \ref{Figure:ecdensity}) and expected EC (Fig. \ref{Figure:avgupcross}) plots.} based on the spatial distance $r$. This results in maximum intracell interference and causes blockage for the cellular user. If we increase $u$ (red, green, and blue transparent planes), only those potential D2D pairs will survive that lie under high intensity blobs of $\chi^2_k$-RF. In this case, the coverage probability of cellular user can be ensured while reusing the resources for D2D pairs. The coverage probability curves for different $u$ and $\gamma$ can be seen in Fig. \ref{Figure:covprob}. The high threshold, for example, $u = 35$ in this figure shows no D2D pair and ensure the unit coverage probability of cellular user. 

The interference control parameter $u$ for different grid sizes has been plotted in Fig. \ref{Figure:gridvsu}. In this figure, it can be seen that the interference, due to D2D communication on coverage probability of cellular user, is captured by $u$. As an example, for 10$\times$10 grid size (SBS of radius $R$ = 5 m), the $p^c_{cov}$ rises from 2\% to 98\% for $u$ = 5 to 16 as compared to 1000$\times$1000 grid size ($R$ = 500 m) where the blockage extends on the floor upto the value of $u$ = 37 and shows 98\% rise in $p^c_{cov}$ at $u$ =  45.
\begin{figure}[!htb]
\centering
\includegraphics[width=1\columnwidth]{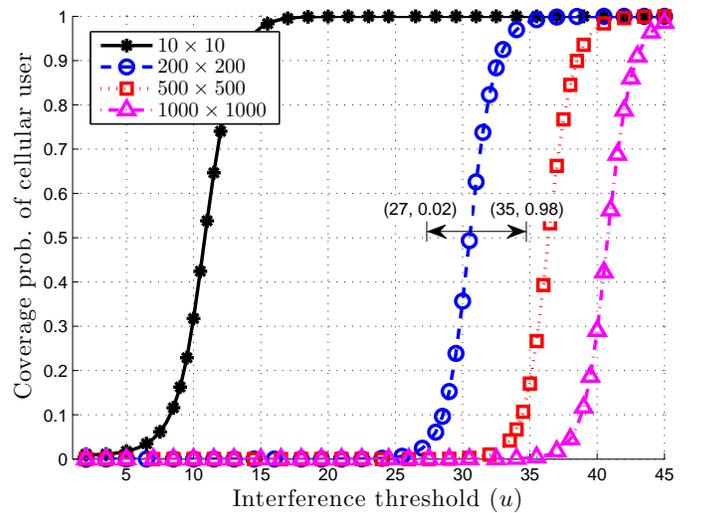}
\caption{Interference control parameter $u$ for different grid sizes (SBS radius $R$) for $p_{i} = $ 0 dBm, $p_{c} = $ 20 dBm.}\label{Figure:gridvsu}
\end{figure}

In Fig. \ref{Figure:clusterinterference}, we have shown the effect of interference due to different cluster sizes on coverage probability of cellular user. The cluster sizes show different number of D2D nodes that survive the threshold $u$. For example, at $u$ = 31 ($p^c_{cov}$ = 0.62 from Fig. \ref{Figure:gridvsu}), four cluster sizes of $r = (16, 8, 4, 2) m$ are shown that consider D2D communication by reusing the frequency of cellular user. The maximum cluster size considers all nodes which are less than or equal to 16 m for D2D communication and hence causes maximum interference. Contrary to this, the minimum cluster size considers nodes with 2m or less distance for D2D communication and hence results in less interference.
\begin{figure}[t]
\centering
\includegraphics[width=1\columnwidth]{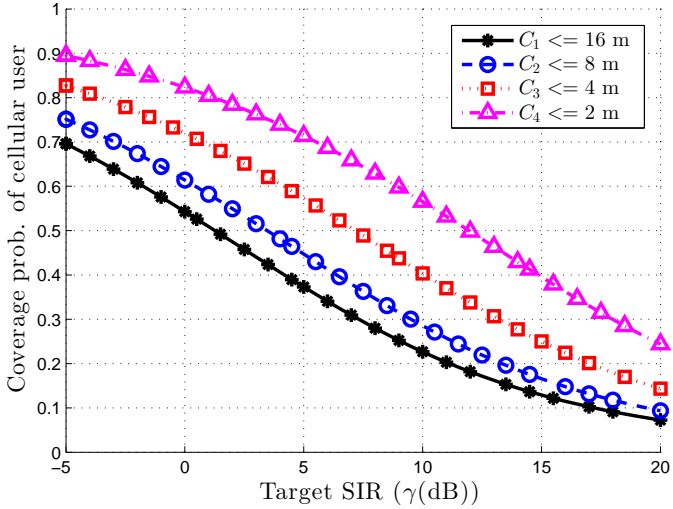}
\caption{Clusters of D2D nodes and coverage probability of cellular user for $u$ = 31, $p_{i} = $ 0 dBm, $p_{c} = $ 20 dBm and $R$ = 100 m.}\label{Figure:clusterinterference}
\end{figure}

The effect of power control on D2D pairs can be seen in Fig. \ref{Figure:variablepi}. In this figure, we can see that the coverage probability of cellular user can be ensured by controlling the transmit power of successful D2D pairs. The coverage drop at two values of $\gamma$ (0 and 20 dB) is approximately equal, however, two curves with smaller $p_i$ ([-20 -10] dBm) converges rapidly at lower values of $\gamma$ as compared to curves with high $p_i$ ([0 10] dBm). This trend is reversed at high values of $\gamma$.

The effect of power control on cellular user and coverage probability curves are shown in Fig. \ref{Figure:variablepc}. It can be seen that the curves for different $p_c$ converge to low coverage probability for high $\gamma$. The coverage probability can be increased by either reducing transmit power of D2D pairs or reducing the number of D2D pairs by increasing threshold $u$. The threshold $u$ and spatial extent $r$ (small $r$ requires lower $p_i$) are key control parameters to ensure the extent of frequency reuse (D2D pairs) while ensuring coverage probability of cellular user.
\section{Conclusion}
In this paper, we proposed PCP process to model inhomogeneous and spatially correlated distribution of MTs. We considered this process to characterize intracell interference in D2D underlay network. We further approximated intractable nearest neighbor distribution function by adopting expected Euler Characteristic and Poisson clumping heuristic. The key findings of this research are enumerated as:
\begin{enumerate}
\item Simple SPPP process with constant intensity measure cannot capture prevailing inhomogeneity and spatial correlation in dense cellular networks. Therefore, point processes with attraction/repulsion property (e.g., Cox process/DPP) are potential candidates for precise spatial modeling of MTs/BSs.
\item Euler Characteristic and RFT framework can be used to analyze and identify high intensity areas/hotspots for D2D communication.
\item Provided SPMs of coverage area are available, statistical inference can be performed to identify clusters of MTs with high spatial correlation (potential areas for D2D communication).
\item The intensity measure of PCP is governed by $\chi^2_k$-RF. In this case the threshold $u$ of the excursion set plays a key role to control cluster size, for D2D communication, level of interference, due to frequency reuse, and coverage probability of cellular user. 
\end{enumerate}

\singlespacing
\section*{APPENDIX A - {Proof of Theorem 4} }\label{AppA}
\renewcommand{\theequation}{A.\arabic{equation}}
The average coverage probability of uplink cellular user distributed uniformly over plane between $R$ and $R_{0}$ at a distance $r_{c}$ from the serving SBS is given as follows:
\begin{align}
\setcounter{equation}{0}
p_{cov}^{c} = & \mathbb{E}_{r_{c}}\big[p\big(\textnormal{SIR}_{SBS} \geq \gamma\big)\,|\,r_{c}\big], \IEEEnonumber
\\ =& \mathbb{E}_{r_{c}}\bigg[p\big(f_{c} \geq \frac{\gamma I_{m}}{p_{c} r^{-\alpha}_{c}}\big)\,|\,r_{c}\bigg],	\label{SIR_AppA}
\end{align}
where 
\begin{align}
I_{m} = \sum_{i \in \Phi} p_{i} f_{i} r^{-\alpha}_{i},
\label{Im}
\end{align}
is the cumulative interference due to D2D clusters in the coverage area and $\mathbb{E}_{(\cdot)}$ is expectation with respect to ($\cdot$).

In (\ref{SIR_AppA}), the coverage probability depends on number of RVs e.g., $p_{c}, f_{c}, r^{-\alpha}_{c}, p_{i}, f_{i}, r^{-\alpha}_{i}$. The power transmitted by the cellular user $p_{c}$ is assumed to be independent of the interferers. The serving SBS uses uplink power control to ensure quality of service of the cellular user based on distance dependent path-loss. The fading $f_{c}$ and $f_i$ follows Rayleigh distribution with $p_c$ and $p_i$ as exponentially distributed. The cellular user is uniformly distributed in the coverage area whereas all potential D2D nodes are distributed according to PCP process. Conditioning on $g = \{p_{i}, f_{i}\}$, the coverage probability of cellular user for a given transmit power $p_c$ is:
\begin{align}
p\big(\textnormal{SIR}_{SBS} \geq \gamma\big)\,|\,r_{c},g =& \int_{x=\frac{\gamma I_{m}}{p_{c} r^{-\alpha}_{c}}}^\infty e^{- x}dx,	\IEEEnonumber
\\ =& e^{-\gamma p^{-1}_{c} r^{\alpha}_{c} I_{m}}.
\label{SIR_AppA1}
\end{align}
De-conditioning by $g$, (\ref{SIR_AppA1}) results in
\begin{align}
p\big(\textnormal{SIR}_{SBS} \geq \gamma\big)\,|\,r_{c} =& \mathbb{E}_{g}\big[e^{-\gamma p^{-1}_{c} r^{\alpha}_{c} I_{m}}\big],	\IEEEnonumber
\\ =& \mathbb{E}_{g}\big[e^{- s_{c} I_{m}}\big],	\IEEEnonumber
\\ =& \mathcal{L}_{I_{m}}\big(s_{c}\big),
\label{SIR_AppA2}
\end{align}
where $s_{c} = \gamma p^{-1}_{c} r^{\alpha}_{c}$. 
\\
Putting the value of $I_{m}$ from (\ref{Im}) in (\ref{SIR_AppA2})
\begin{align}
\mathcal{L}_{I_{m}}\big(s_{c}\big) = & \mathbb{E}_{\Phi,p_{i},f_{i}}\bigg[ e^{-s_{c} \sum_{i \in \Phi} p_{i} f_{i} r^{-\alpha}_{i}}\bigg]	\IEEEnonumber
\\  = & \mathbb{E}_{\Phi,p_{i},f_{i}}\bigg[\prod_{i \in \Phi} e^{-s_{c}  p_{i} f_{i} r^{-\alpha}_{i}}\bigg]	\IEEEnonumber
\\ = & \mathbb{E}_{\Phi}\bigg[\prod_{i \in \Phi} \mathbb{E}_{p_{i}} \bigg(\frac{1}{1+s_{c}  p_i r^{-\alpha}_{i}}\bigg)\bigg]	\IEEEnonumber
\\ = & \mathbb{E}_{\Phi}\bigg[\prod_{i \in \Phi} \underbrace{\bigg(\frac{1}{1+s_{c} \mathbb{E} [p_i] r^{-\alpha}_{i}}\bigg)}_{f(x)}\bigg],
\label{SIR_AppA3}
\end{align}
where (\ref{SIR_AppA3}) results from the i.i.d. assumption of $p_{i}$ and $f_{i}$ and further independence from PCP process.
\begin{figure}[t]
\centering
\includegraphics[width=1\columnwidth]{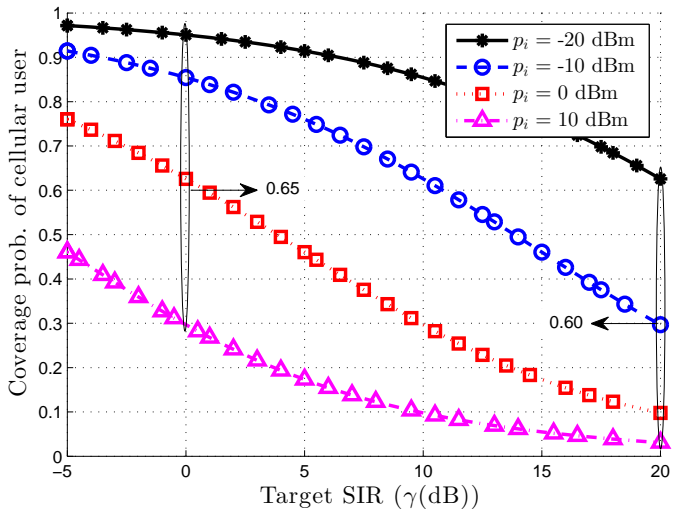}
\caption{Power control on D2D pairs and effect on coverage probability of cellular user for $u$ = 31, $p_{c} = $ 20 dBm and $R$ = 100 m.}\label{Figure:variablepi}
\end{figure}

The probability generating functional of $f(x)$ with retention probability $p(r) = G(r)$ from (\ref{pumc}) and $\Lambda$ from (Sec. \ref{appim}) is\footnote{The Laplace functional of SPPP, with random intensity measure $\Lambda$ and thinning based on $G$ function of PCP has been applied \cite[Proposition 1.2.2 and 1.3.5]{blaszczyszyn_stochastic_2009} to capture interference due to successful D2D clusters.}
\begin{align}
\mathbb{E}_{\Phi}\bigg[\prod_{i \in \Phi}f(x) \bigg] = & \,e^{- \int_{\mathbb{R}^2} (1-f(x)) p(r) \psi_0 dx},	\IEEEnonumber
\\ = & \,e^{-\psi_0\, p(r) \int^{\infty}_{0} \int^{2\pi}_{0} (1-f(x)) x d\theta dx},	\IEEEnonumber
\\ = & \,e^{-2\pi \psi_0\, p(r) \int^{\infty}_{0} (1-f(x)) x dx}.
\label{PGFL_AppA}
\intertext{Putting $f(x)$ from (\ref{SIR_AppA3}) in (\ref{PGFL_AppA}), $\mathcal{L}_{I_{m}}(\cdot)$ can be computed as}
\mathcal{L}_{I_{m}}\big(s_{c}\big) = & e^{-2\pi \psi_0\, p(r) \int^{\infty}_{R_0} \big(1- \frac{1}{1+s_{c}\mathbb{E}[p_{i}] x^{-\alpha}} \big) x dx},	\IEEEnonumber
\\ = & e^{-2\pi \psi_0\, p(r) \int^{\infty}_{R_0} \big(\frac{1}{1+\frac{ x^{\alpha}}{s_{c}\mathbb{E}[p_{i}]}} \big) x dx}.
\label{PGFL_AppA1}
\intertext{By substituting $\frac{x^{\alpha}}{s_{c}\mathbb{E}[p_{i}]} = u^{\alpha}$, (\ref{PGFL_AppA1}) can be derived as}
\mathcal{L}_{I_{m}}\big(s_{c}\big) = & e^{-2\pi \psi_0 \, p(r) (s_{c})^{\frac{2}{\alpha}} \mathbb{E}[{p^{\frac{2}{\alpha}}_{i}}] \int^{\infty}_{R_0} \big(\frac{u}{1+u^{\alpha}}\big) du}.
\label{PGFL_AppA2}
\end{align}
Since $R_0 \ll R$, therefore assuming $R_0 \sim 0$, the integral on right hand side of (\ref{PGFL_AppA2}) can be evaluated as:
\begin{align}
\int^{\infty}_{0} \bigg(\frac{u}{1+u^{\alpha}}\bigg) du = & \frac{\pi }{\alpha \sin(\frac{2 \pi}{\alpha})}.
\label{PGFL_AppA3}
\end{align}
Putting (\ref{PGFL_AppA3}) in (\ref{PGFL_AppA2}) and using uniform distribution from (\ref{frc}), the average coverage probability of a cellular user (\ref{SIR_AppA}) is:
\begin{align}
p_{cov}^{c} = & \mathbb{E}_{r_{c}}\bigg[e^{-\frac{2\pi^2 \psi_0 \, p(r) \,r^{2}_{c}}{\alpha \sin(\frac{2 \pi}{\alpha})} \big(\frac{\gamma}{ p_{c}}\big)^{\frac{2}{\alpha}} \mathbb{E}[{p^{\frac{2}{\alpha}}_{i}}]}\,|r_{c}\bigg], \IEEEnonumber
\\ =& \int^{R}_{R0}e^{-\frac{2\pi^2 \psi_0 p(r)\,r^{2}_{c}}{\alpha \sin(\frac{2 \pi}{\alpha})} \big(\frac{\gamma}{ p_{c}}\big)^{\frac{2}{\alpha}} \mathbb{E}[{p^{\frac{2}{\alpha}}_{i}}] }\,\frac{2 r_{c}}{R^{2}} dr_{c}.
\label{SIR_AppA5}
\end{align}
\begin{figure}[t]
\centering
\includegraphics[width=1\columnwidth]{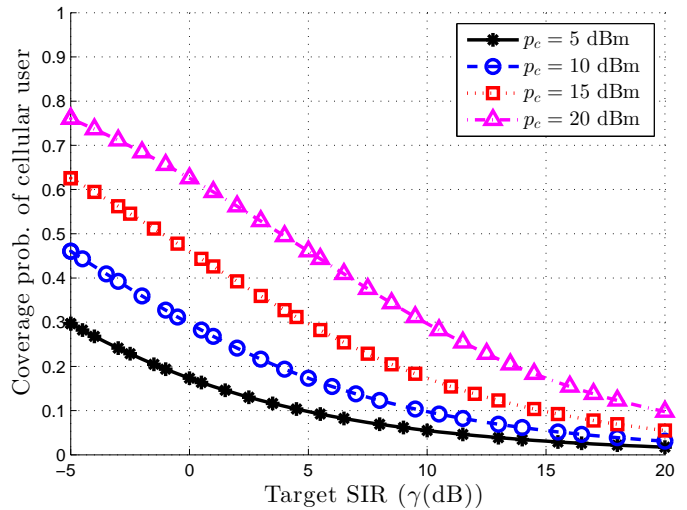}
\caption{Power control on cellular user and coverage probability for $u$ = 31, $p_{i} = $ 0 dBm and $R$ = 100 m.}\label{Figure:variablepc}
\end{figure}

\bibliographystyle{IEEEtran}
\bibliography{IEEEabrv,ref_Cox}
\end{document}